\documentclass[11pt,reqno]{amsart}
\usepackage[utf8]{inputenc}
\usepackage{latexsym,amsmath}
\usepackage[left=3cm,top=2.5cm,right=3cm,bottom=2.5cm]{geometry}
\usepackage{color}
\usepackage{amsthm}
\usepackage{amssymb}
\usepackage{epsfig}

\begin{document}

\pagestyle{empty}
\newtheorem{theorem}{Theorem}[section]
\newtheorem{cor}[theorem]{Corollary}
\newtheorem{lemma}[theorem]{Lemma}
\newtheorem{fact}[theorem]{Fact}
\newtheorem{property}[theorem]{Property}
\newtheorem{proposition}[theorem]{Proposition}
\newtheorem{claim}[theorem]{Claim}
\newtheorem{conjecture}[theorem]{Conjecture}
\newtheorem{question}[theorem]{Question}
\newtheorem{definition}[theorem]{Definition}
\theoremstyle{definition}
\newtheorem{example}[theorem]{Example}
\newtheorem{remark}[theorem]{Remark}

\newcommand\eps{\varepsilon}

\def\St{{\mathcal S}}
\def\K{{\mathcal C}}
\def\R{{\mathcal R}}

\title{Fixed parameter algorithms for restricted coloring problems}

\author[V. Campos]{Victor Campos}
\address{Universidade Federal do Cear\'a, Fortaleza, Brazil}
\email{\tt campos@lia.ufc.br}

\author[C. Linhares-Sales]{Cl\'audia Linhares-Sales}
\address{Universidade Federal do Cear\'a, Fortaleza, Brazil}
\email{\tt linhares@lia.ufc.br}

\author[A. K. Maia]{Ana Karolinna Maia}
\address{Universidade Federal do Cear\'a, Fortaleza, Brazil}
\email{\tt karol@lia.ufc.br}

\author[N. Martins]{Nicolas Martins}
\address{Universidade Federal do Cear\'a, Fortaleza, Brazil}
\email{\tt nicolasam@lia.ufc.br}

\author[R. Sampaio]{Rudini M. Sampaio}
\address{Universidade Federal do Cear\'a, Fortaleza, Brazil}
\email{\tt rudini@ufc.br}

\thanks{The authors are partially supported by Funcap (Proc. 07.013.00/09) and CNPq (Proc.~484154/2010-9)}

\date{}
\maketitle

\begin{abstract}
In this paper, we obtain polynomial time algorithms to determine the acyclic chromatic number, the star chromatic number, the Thue chromatic number, the harmonious chromatic number and the clique chromatic number of $P_4$-tidy graphs and $(q,q-4)$-graphs, for every fixed $q$. These classes include cographs, $P_4$-sparse and $P_4$-lite graphs. All these coloring problems are known to be NP-hard for general graphs. These algorithms are fixed parameter tractable on the parameter $q(G)$, which is the minimum $q$ such that $G$ is a $(q,q-4)$-graph. We also prove that every connected $(q,q-4)$-graph with at least $q$ vertices is 2-clique-colorable and that every acyclic coloring of a cograph is also nonrepetitive.
\end{abstract}



\section{Introduction}

Let $G=(V,E)$ be a finite undirected graph, without loops and multiple edges. The complete bipartite graph with partitions of size $m$ and $n$ is denoted by $K_{m,n}$. A $K_{1,n}$ is called a star. A $P_4$ is an induced path with four vertices. A cograph is any $P_4$-free graph. The graph terminology used here follows {\cite{BM08}}.

A $k$-coloring of $G$ is a partition $\{V_1,\ldots,V_k\}$ of $V(G)$. The subsets $V_1,\ldots,V_k$ are called \emph{color classes} and we say that a vertex in $V_i$ is colored $i$. A proper $k$-coloring is a $k$-coloring such that every color class induces a stable set. The chromatic number $\chi(G)$ of $G$ is the smallest integer $k$ such that $G$ admits a proper $k$-coloring.

An \emph{acyclic} coloring is a proper coloring such that every cycle receives at least three colors (that is, every pair of color classes induces a forest).
A \emph{star} coloring is an acyclic coloring such that every $P_4$ receives at least three colors (that is, every pair of color classes induces a forest of stars).
A \emph{nonrepetitive} coloring is a star coloring such that no path has a $xx$ pattern of colors, where $x$ is a sequence of colors.
A \emph{harmonious} coloring is a nonrepetitive coloring such that every pair of color classes induces at most one edge.

It is easy to see that any coloring of a split or chordal graph is acyclic. In 2011, Lyons \cite{lyons11} proved that every acyclic coloring of a cograph is also a star coloring. In this paper, we prove that it is also nonrepetitive.

The acyclic, star, Thue and harmonious chromatic numbers of $G$, denoted respectively by $\chi_a(G)$, $\chi_{st}(G)$, $\pi(G)$, $\chi_h(G)$, are the minimum number of colors $k$ such that $G$ admits an acyclic, star, nonrepetitive and harmonious coloring with $k$ colors. By definitions, $$\chi(G)\ \leq\ \chi_a(G)\ \leq\ \chi_{st}(G)\ \leq\ \pi(G)\ \leq\ \chi_h(G).$$

Determining the acyclic chromatic number is NP-Hard even for bipartite graphs \cite{coleman86} and deciding if $\chi_a(G)\leq 3$ is NP-Complete \cite{kostochka78}.
In 2004, Albertson et al. \cite{albertson04} proved that computing the star chromatic number is NP-hard even for planar bipartite graphs.
In 2007, Asdre et al. \cite{asdre07} proved that determining  the harmonious chromatic number is NP-hard for interval graphs, permutation graphs and split graphs.

Borodin proved that $\chi_a(G)\leq 5$ for every planar graph $G$ \cite{borodin79}.
In 2004, Fertin, Raspaud and Reed give exact values of $\chi_{st}(G)$ for several graph classes \cite{reed04}. In 2004, Campbell and Edwards \cite{campbell04} obtained new lower bounds for $\chi_h(G)$ in terms of the independence number.

In 2002, Alon et al. \cite{alon02} proved a relation between the $\pi(G)$ and $\Delta(G)$.
In 2008, Barát and Wood \cite{barat08} proved that every graph $G$ with treewidth $t$ and maximum degree $\Delta$ satisfies $\pi(G)=O(k\Delta)$ (it was also proved that $\pi(G)\leq 4^t$ \cite{gryt07}). 
In 2009, Marx and Schaefer \cite{marx09} proved that determining whether a particular coloring of a graph is nonrepetitive is coNP-hard, even if the number of colors is limited to four.
In 2010, Grytczuk et al. \cite{gryt10} investigated list colorings which are nonrepetitive and proved that the Thue choice number of $P_n$ is at most 4 for every $n$. See \cite{gryt07} and \cite{gryt08} for a survey on nonrepetitive colorings.

A \emph{clique} coloring is a coloring (not necessarily a proper coloring) such that every maximal clique receives at least two colors. The clique chromatic number $\chi_c(G)$ is the minimum number $k$ such that $G$ has a clique coloring with $k$ colors.

In 2002, Kratochv\'il and Tuza \cite{kratochvil02} proved that determining the clique-chromatic number is polynomial time solvable for planar graphs, but is NP-Hard for perfect graphs. In 2004, Bacs\'o et al. \cite{bacso04} proved several results for 2-clique-colorable graphs.  

Many NP-hard problems were proved to be polynomial time solvable for cographs. For example, Lyons \cite{lyons11} obtained a polynomial time algorithm to find an optimal acyclic and an optimal star coloring of a cograph. However, it is known that computing the harmonious chromatic number of a disconnected cograph is NP-hard \cite{bodlaender89}.

Some superclasses of cographs, defined in terms of the number and structure of its induced $P_4$'s, can be completely characterized by their modular or primeval decomposition. Among these classes, we cite $P_4$-sparse graphs, $P_4$-lite graphs, $P_4$-tidy graphs and $(q,q-4)$-graphs.

Babel and Olariu \cite{olariu98} defined a graph as $(q,q-4)$-\emph{graph} if no set of at most $q$ vertices induces more than $q-4$ distinct $P_4$'s. Cographs and $P_4$-sparse graphs are precisely $(4,0)$-graphs and $(5,1)$-graphs respectively. $P_4$-lite graphs are special $(7,3)$-graphs.
We say that a graph is $P_4$-tidy if, for every $P_4$ induced by $\{u,v,x,y\}$, there exists at most one vertex $z$ such that $\{u,v,x,y,z\}$ induces more than one $P_4$.
Since the complement of a $P_4$ is also a $P_4$, these graph classes are closed under complementation.

In this paper, we prove the following result:

\begin{theorem}[main theorem]\label{teo-main1}
Let $q$ be a fixed integer and let $G$ be a $P_4$-tidy or a $(q,q-4)$-graph. There exists linear time algorithms to obtain
\begin{itemize}
\item $\chi_a(G)$, $\chi_{st}(G)$, $\pi(G)$ and $\chi_c(G)$;
\item $\chi_h(G)$, if $G$ is also connected.
\end{itemize}
Moreover, every connected $(q,q-4)$-graph with at least $q$ vertices is 2-clique-colorable, and every acyclic coloring of a cograph is also nonrepetitive.
\end{theorem}

Let $q(G)$ be the minimum integer $q$ such that $G$ is a $(q,q-4)$-graph. Theorem \ref{teo-main1} proves that the acyclic, the star, the nonrepetitive, the harmonious and the clique coloring problems are fixed parameter tractable on the parameter $q(G)$.

\section{Primeval and Modular decompositions}\label{secao-primeval}

Let $G_1=(V_1,E_1)$ and $G_2=(V_2,E_2)$ be two vertex disjoint graphs.
The disjoint union of $G_1$ and $G_2$ is the graph $G_1\cup G_2=(V_1\cup V_2,E_1\cup E_2)$. The join is the graph $G_1\vee G_2=(V_1\cup V_2,E_1\cup E_2\cup \{uv:\ u\in V_1,\ v\in V_2\})$. 

A \emph{spider} is a graph whose vertex set has a partition $(R,C,S)$, where $C=\{c_1,\ldots,c_k\}$ and $S=\{s_1,\ldots,s_k\}$ for $k\geq 2$ are respectively a 
clique and a stable set; $s_i$ is adjacent to $c_j$ if and only if $i=j$ (a thin spider), or $s_i$ is adjacent to $c_j$ if and only if $i\not=j$ (a thick spider); and every vertex of $R$ is adjacent to each vertex of $C$ and non-adjacent to each vertex of $S$. 


Jamison and Olariu \cite{olariu98} proved an important structural theorem for $(q,q-4)$-graphs, using their primeval decomposition, which can be obtained in linear time.
A graph is {\it $p$-connected} if, for every bipartition of the vertex set, there is a crossing $P_4$. A {\it separable p-component} is a maximal p-connected subgraph with a particular bipartition $(H_1,H_2)$ such that every crossing $P_4$ $wxyz$ satisfies $x,y\in H_1$ and $w,z\in H_2$.

\begin{theorem}[Characterizing $(q,q-4)$-graphs \cite{olariu98}]\label{teo-primeval}
A graph $G$ is a $(q,q-4)$-graph if and only if exactly one of the following holds:
\begin{itemize}
\item [(a)] $G$ is the union or the join of two $(q,q-4)$-graphs;
\item [(b)] $G$ is a spider $(R,C,S)$ and $G[R]$ is a $(q,q-4)$-graph; 
\item[(c)] $G$ contains a separable $p$-component $H$, with bipartition $(H_1,H_2)$ and $|V(H)| \leq q$, such that $G-H$ is a $(q,q-4)$-graph and every vertex of $G-H$ is adjacent to every vertex of $H_1$ and non-adjacent to every vertex of $H_2$;
\item[(d)] $G$ has at most $q$ vertices or $V(G)=\emptyset$.
\end{itemize}
\end{theorem}


Using the modular decomposition of $P_4$-tidy graphs, Giakoumakis et al. proved a similar result for this class \cite{giakoumakis97}. A \emph{quasi-spider} is a graph obtained from a spider $(R,C,S)$ by replacing at most one vertex from $C\cup S$ by a $K_2$ (the complete graph on two vertices) or a $\overline{K_2}$ (the complement of $K_2$).

\begin{theorem}[Characterizing $P_4$-tidy graphs \cite{giakoumakis97}]\label{teo-primeval2}
A graph $G$ is a $P_4$-tidy graph if and only if exactly one of the following holds:
\begin{itemize}
\item [(a)] $G$ is the union or the join of two $P_4$-tidy graphs;
\item [(b)] $G$ is a quasi-spider $(R,C,S)$ and $G[R]$ is a $P_4$-tidy graph; 
\item [(c)] $G$ is isomorphic to $P_5$, $\overline{P_5}$, $C_5$, $K_1$ or $V(G)=\emptyset$.
\end{itemize}
\end{theorem}

As a consequence, a $(q,q-4)$-graph (resp. a $P_4$-tidy graph) $G$ can be decomposed by successively applying Theorem \ref{teo-primeval} 
(resp. Theorem \ref{teo-primeval2}) as follows:
If (a) holds, apply the theorem to each component of $G$ or $\overline{G}$ (operations disjoint union and join). If (b) holds, 
apply the theorem to $G[R]$ (operation spider or quasi-spider). Finally, if (c) holds and $G$ is a $(q,q-4)$-graph,
then apply the theorem to $G-H$ (operation small subgraph).

It was also proved in \cite{olariu98} that every p-connected $(q,q-4)$-graph with $q\geq 8$ has at most $q$ vertices. With this, we can obtain $q(G)$ in $O(n^7)$ time for every graph $G$ from its primeval decomposition (observe that $q(G)$ can be greater than $n$ and, if this is the case, $q(G)$ is the number of induced $P_4$'s of $G$ plus four).

The idea now is to consider the graph by the means of its decomposition tree obtained as described. According to the coloring parameter to
be determined, the tree will be visited in an up way or bottom way fashion.
We notice that the primeval and modular decomposition of any graph can be obtained in linear time \cite{olariu98}.

\section{Disjoint Union, Join and Spiders}\label{secao-operacoes}

We start by recalling a result from \cite{lyons11} for the acyclic and the star chromatic numbers.

\begin{lemma}[$\chi_a$ and $\chi_{st}$ for union and join \cite{lyons11}]\label{teo-lyons}
Given graphs $G_1$ and $G_2$ with $n_1$ and $n_2$ vertices res-pectively:
$$\chi_a(G_1\cup G_2)=\max\{\chi_a(G_1),\chi_a(G_2)\},\ \ \ \ \ \ \ \ \ \ \ \ \ \ \ $$
$$\chi_{st}(G_1\cup G_2)=\max\{\chi_{st}(G_1),\chi_{st}(G_2)\},\ \ \ \ \ \ \ \ \ \ \ \ \ $$
$$\chi_a(G_1\vee G_2)=\min\{\chi_a(G_1)+n_2,\chi_a(G_2)+n_1\},\ \ $$
$$\chi_{st}(G_1\vee G_2)=\min\{\chi_{st}(G_1)+n_2,\chi_{st}(G_2)+n_1\}.$$
\end{lemma}

The next lemma shows how to obtain the Thue chromatic number for union and join operations. It is easy to see that Lemmas \ref{teo-lyons} and \ref{lema-nonrep} implies that, if $G$ is a cograph, then $\pi(G)=\chi_{st}(G)=\chi_a(G)$ and every acyclic coloring of a cograph is also nonrepetitive.

\begin{lemma}[$\pi(G)$ for union and join]\label{lema-nonrep}
Given graphs $G_1$ and $G_2$ with $n_1$ and $n_2$ vertices res-pectively:
$$\pi(G_1\cup G_2)=\max\{\pi(G_1),\pi(G_2)\},\ \ \ \ \ \ \ \ \ \ \ \ \ $$
$$\pi(G_1\vee G_2)=\min\{\pi(G_1)+n_2,\pi(G_2)+n_1\}.$$
\end{lemma}

The two following lemmas deal with spiders and quasi-spiders and are proved in Section \ref{secao-tecnica}.
We will consider $\chi_a(G[R])=\chi_{st}(G[R])=0$ whenever $R=\emptyset$.

\begin{lemma}[$\chi_a$ and $\chi_{st}$ for spiders]\label{lema-spid-as}
Let $G$ be a spider $(R,C,S)$, where $|C|=|S|=k$.
Then $\chi_a(G)=\chi_a(G[R])+k$ and $\chi_{st}(G)=\chi_{st}(G[R])+k$, unless
$R=\emptyset$ and $G$ is thick, when in this case, $\chi_{st}(G)=k+1$. Moreover, $\pi(G)=\chi_{st}(G)$.
\end{lemma}

\begin{lemma}[$\chi_a$ and $\chi_{st}$ for quasi-spiders]\label{lema-qspid-as}
Let $G$ be a quasi-spider $(R,C,S)$ such that $\min\{|C|,|S|\}=k$ and $\max\{|C|,|S|\}=k+1$. Let $H=K_2$ or $H=\overline{K_2}$ be the subgraph that replaced a vertex of $C \cup S$. Then
\[
  \chi_a(G)\ =\ 
  \begin{cases}
                \chi_a(G[R])+k+1,&\mbox{if $H\in C$},\\
                \chi_a(G[R])+k+1,&\mbox{if $H=K_2$, $G$ is thick}\\
                                 &\ \ \ \ \ \ \ \ \ \ \ \ \ \mbox{and $R=\emptyset$},\\
                \chi_a(G[R])+k,  &\mbox{otherwise},               
  \end{cases}
\]
\[
  \chi_{st}(G)\ =\ 
  \begin{cases} \chi_{st}(G[R])+k,  &\mbox{if $H\in S$ and $G$ is thin},\\
                \chi_{st}(G[R])+k,  &\mbox{if $H\in S$, $G$ is thick}\\
                                 &\ \ \ \ \ \ \ \ \ \ \ \ \ \mbox{and $R\not=\emptyset$},\\
                \chi_{st}(G[R])+k+2,&\mbox{if $H\in C$, $G$ is thick}\\
                                 &\ \ \ \ \ \ \ \ \ \ \ \ \ \mbox{and $R=\emptyset$},\\
                \chi_{st}(G[R])+k+1,&\mbox{otherwise}.
  \end{cases}
\]
Moreover, $\pi(G)=\chi_{st}(G)$.
\end{lemma}

Lemma below determines the harmonious chromatic number for join and spider operations. Recall that $\chi_h$ for union operation is NP-hard {\cite{bodlaender89}}.

\begin{lemma}[$\chi_h$ for join and quasi-spiders]\label{lema-spid-h}
Let $G$ be a graph with $n$ vertices.
If $G$ is the join of two graphs $G_1$ and $G_2$, then $\chi_h(G)=n$.
If $G$ is a quasi-spider $(R,C,S)$ with $k=\max\{|C|,|S|\}$, then
\[
  \chi_h(G)\ =\ 
  \begin{cases} |R|+k+1, &\mbox{if $G$ is thin},\\
                n,      &\mbox{otherwise.}
  \end{cases}
\]
\end{lemma}

\begin{lemma}[$\chi_c$ for union, join and quasi-spiders]\label{lema-spid-c}
Let $G_1$ and $G_2$ be two graphs. Then, $\chi_c(G_1\cup G_2)=\max\{\chi_c(G_1),\chi_c(G_2)\}$ and $\chi_c(G_1\vee G_2)=2$.
If $G$ is a quasi-spider, then $\chi_c(G)=2$.
\end{lemma}

\section{Coloring $(q,q-4)$-graphs}\label{secao-small}

In this section, suppose that $G$ is a $(q,q-4)$-graph which contains a separable $p$-component $H$, with bipartition $(H_1,H_2)$ and at most $q$ vertices, such that every vertex from $G-H$ is adjacent to all vertices in $H_1$ and non-adjacent to all vertices in $H_2$. Let $n'$ be the number of vertices of $G-H$. If $G-H$ is empty, consider $\chi_a(G-H)=\chi_{st}(G-H)=\pi(G-H)=0$. Given a coloring $\psi$ of $H$, let $k(\psi)$ be the number of colors of $\psi$.

Theorems below prove that determining the chromatic numbers $\chi_a$, $\chi_{st}$, $\chi_h$ and $\chi_c$ for item (c) of Theorem \ref{teo-primeval} is linear time solvable, if $q$ is a fixed integer.

\begin{lemma}\label{teo-small-as}
Given a coloring $\psi$ of $H$, let $k_2(\psi)$ be the number of colors with no vertex of $H_1$ and with no vertex of $H_2$ which is neighbor of two vertices from $H_1$ with the same color. Then
\[
  \chi_a(G)\ =\ \min\Big\{\min_{\psi\in C_a(H)}\Big\{k(\psi)+\max\{0,n'-k_2(\psi)\}\Big\},
\]
\[\ \ \ \ \ \ \ \ \ \ \ \ \ \ \ \ \ \ \ \ \ \ \ \ \ \ \ \ \ \ \ \ \ 
  \min_{\psi'\in C'_a(H)}\Big\{k(\psi')+\max\{0,\chi_a(G-H)-k_2(\psi')\}\Big\}\Big\}
\]
\[
  \chi_{st}(G)\ =\ \min\Big\{\min_{\psi\in C_{st}(H)}\Big\{k(\psi)+\max\{0,n'-k_2(\psi)\}\Big\},
\]
\[\ \ \ \ \ \ \ \ \ \ \ \ \ \ \ \ \ \ \ \ \ \ \ \ \ \ \ \ \ \ \ \ \ 
  \min_{\psi'\in C'_{st}(H)}\Big\{k(\psi')+\max\{0,\chi_{st}(G-H)-k_2(\psi')\}\Big\}\Big\}
\]
\[
  \pi(G)\ =\ \min\Big\{\min_{\psi\in C_\pi(H)}\Big\{k(\psi)+\max\{0,n'-k_2(\psi)\}\Big\},
\]
\[\ \ \ \ \ \ \ \ \ \ \ \ \ \ \ \ \ \ \ \ \ \ \ \ \ \ \ \ \ \ \ \ \ 
  \min_{\psi'\in C'_\pi(H)}\Big\{k(\psi')+\max\{0,\pi(G-H)-k_2(\psi')\}\Big\}\Big\}
\]
where $C_a(H)$, $C_{st}(H)$ and $C_\pi(H)$ are respectively the set of all acyclic, star and nonrepetitive colorings of $H$, and $C'_a(H)\subseteq C_a(H)$, $C'_{st}(H)\subseteq C_{st}(H)$ and $C'_\pi(H)\subseteq C_\pi(H)$ are respectively the subsets of acyclic, star and nonrepetitive colorings such that all vertices from $H_1$ receive distinct colors.
\end{lemma}

\begin{lemma}\label{teo-small-h}
Let $C_h(H)$ be the set of all harmonious colorings of $H$ such that all vertices of $H_1$ have distinct colors. Then
$$
  \chi_h(G)\ =\ n'+\min_{\psi\in C_h(H)}\Big\{k(\psi)\Big\}
$$
\end{lemma}

\begin{lemma}\label{teo-small-c}
If $G-H$ is not empty, then $\chi_c(G)=2$ (coloring the vertices of $G-H$ and $H_2$ with the color 1 and the vertices of $H_1$ with the color 2). If $G-H$ is empty, then $G$ has less than $q$ vertices and
$$
  \chi_c(G)\ =\ \min_{\psi\in C_c(H)}\Big\{k(\psi)\Big\},
$$
where $C_c(H)$ is the set of all clique-colorings of $H$.
\end{lemma}



\begin{theorem}\label{teo-small}
If $G$ is a $P_4$-tidy or $(q,q-4)$-graph, then we can obtain a minimum acyclic-star-harmonious-clique coloring of $G$ and determine $\chi_a(G)$, $\chi_{st}(G)$, $\chi_h(G)$ and $\chi_c(G)$ in linear time.
\end{theorem}

\begin{proof}
From Section \ref{secao-primeval}, we can obtain the primeval decomposition in linear time.
From lemmas of Sections \ref{secao-operacoes} and \ref{secao-small}, we are finished.
\end{proof}

\section{Technical proofs}\label{secao-tecnica}

We now provide the proofs of the most important results of the paper.
Firstly, we need to state a definition and recall a theorem from \cite{olariu01}.

\begin{definition}\label{def-modulo}
Let $G=(V,E)$ be a graph. A subset $M$ of $V$ with $1\leq|M|\leq|V|$ is called a \emph{module} if each vertex in $V - M$ is either adjacent to all vertices of $M$ or to none of them. A module $M$ is called a \emph{homogeneous set} if $1<|M|<|V|$. The graph obtained from $G$ by shrinking every maximal homogeneous set to one single vertex is called the \emph{characteristic graph} of $G$.
\end{definition}

A graph is called \emph{split graph} if its vertex set has a partition $(K,S)$ such that $K$ induces a clique and $S$ induces an independent set.

\begin{lemma}[\cite{olariu01}]\label{lemma-split}
A p-connected graph $G$ is separable if and only if its characteristic graph is a split graph.
\end{lemma}

The rest of the paper is dedicated to prove Theorem \ref{teo-small} and lemmas from Sections \ref{secao-primeval} and \ref{secao-operacoes}.

\subsection{Acyclic, star and nonrepetitive colorings}

We start with the proofs of Lemmas \ref{lema-nonrep}, \ref{lema-spid-as} and \ref{lema-qspid-as}.

\begin{proof}[Proof of Lemma \ref{lema-nonrep}]
If $G=G_1\cup G_2$, then every color of $G_1$ can be used in $G_2$, and vice-versa.
Thus, $\pi(G)=\max\{\pi(G_1),\pi(G_2)\}$.
So, let $G=G_1\vee G_2$. Suppose that $|V(G_1)|\geq 2$ and $|V(G)|\geq 2$. Let $a_1,b_1\in V(G_1)$ and $a_2,b_2\in V(G_2)$.
Suppose that $a_1$ and $b_1$ receive color $C_1$ and that $a_2$ and $b_2$ receive color $C_2$. Then we have the bicolored $P_4$ $a_1a_2b_1b_2$, which is the repetition pattern $C_1C_2C_1C_2$; a contradiction. So, (a) all vertices of $G_1$ have distinct colors; or (b) all vertices of $G_2$ have distinct colors.
\end{proof}

\begin{proof}[Proof of Lemma \ref{lema-spid-as}]
Let $G$ be a spider with partition $(R,C,S)$, such that $|C|=|S|=k$. 
A minimum acyclic coloring of $G$ can be easily obtained from an acyclic coloring of $G[R]$, by assigning a new color for each vertex in $C$ and finally by coloring each vertex of $S$ with any appropriated avaiable color of $C$. Thus, $\chi_a(G)=\chi_a(G[R])+k$.

On the other hand, to produce a star coloring of $G$, we first color optimally $G[R]$ and then assign one new color to each vertex of $C$.
If $G$ is thin, we color each vertex of $S$ with any appropriated available color of $C$.
If $G$ is thick and $R\not=\emptyset$, then we use one of the colors of $R$ to color every vertex of $S$. Then, $\chi_{st}(G)=\chi_{st}(R)+k$.
If $G$ is thick and $R=\emptyset$, then we have to add a new color and assign it to every vertex of $S$. By consequence, $\chi_{st}(G)=k+1$.
The same arguments can be used to $\pi(G)$.
\end{proof}

\begin{proof}[Proof of Lemma \ref{lema-qspid-as}]
Let $G$ be a quasi-spider with partition $(R,C,S)$ and $\min\{|C|,|S|\}=k$ and $\max\{|C|,|S|\}=k+1$.
Let $H=K_2$ or $H=\overline{K_2}$ be the subgraph that replaced a vertex of $C\cup S$ in the definition of quasi-spider.
As before, we obtain a minimum acyclic(star) coloring of $G$ by an acyclic(star) coloring $R$ and associating a new color to each vertex of $C$.
If $H=K_2$ and $H\in C$, then we need $k+1$ colors to color the vertices of $C$.
If $H=\overline{K_2}$, $H\in C$ and we use only one vertex to color $H$, then we have to give a new color to at least one vertex of $S$, otherwise we would have a bichromatic $C_4$ with $H$, this vertex of $S$ and a vertex of $R$ or $C$.

Then, if $H\in C$, then it is better to use $k+1$ colors on $C$ (independently if $H=K_2$ or $H=\overline{K_2}$). We will assume this.
If $H\not\in C$, then we use $k$ colors on $C$.
Now, we have to color the vertices of $S$.

At first, if $H\in C$, we can color each vertex of $S$ using a color of some vertex in $C$, without producing any bichromatic cycle (but we can produce a bichromatic $P_4$). Then, if $H\in C$, we need no further new colors to acyclic color $S$.

If $H\in S$ and $G$ is thin, we can always use a color of $C$ to color each vertex of $S$ without producing a bichromatic $P_4$.
If $H\in S$, $G$ is thick and $R\not=\emptyset$, we can use one color of $R$ to color all vertices of $S$ without producing a bichromatic $P_4$.
In these cases, we do not need further new colors to color $S$.

If $H=K_2\in S$, $G$ is thick and $R=\emptyset$, we have create a new color and give it to all vertices of $S$ (except one vertex of $H=K_2\in S$, whose color could be some color used in $C$), without producing a bichromatic $P_4$. In this case, we need only one new color. If $H=\overline{K_2}$, we can use a color of $C$.

If $H\in C$, $G$ is thick and $R=\emptyset$, we have to create a new color and  give it to all vertices of $S$. In this case, we need only one new color.
And we have finished all possibilities.
The same arguments can be used to $\pi(G)$.
\end{proof}

\begin{proof}[Proof of Lemma \ref{teo-small-as}]
Clearly, an acyclic-star coloring of $G$ induces an acyclic-star coloring of $G-H$.
Notice that if $x\in G-H$ and $v\in H_2$ have the same color and if $v$ has two neighbors $u_1$ and $u_2$ with the same color, then we have a bichromatic $C_4$ ($x-u_1-v-u_2$). So, this cannot happen in an acyclic coloring of $G$.

Also observe that, in an acyclic coloring of $G$, there not exist vertices $x,y\in G-H$ with the same color and vertices $u,v\in H_1$ with the same color (otherwise, we would have a bichromatic $C_4$ $x-u-y-v$). Therefore, we have two options: (a) each vertex of $G-H$ receives a distinct color or (b) each vertex of $H_1$ receives a distinct color. Firstly, assume (b).

Observe that there is no bichromatic $P_4$ with three vertices of $G-H$ and a vertex of $H_1$, since every vertex of $G-H$ is adjacent to every vertex of $H_1$. Furthermore, there is no bichromatic $P_4$ with one vertex of $G-H$ and two vertices of $H_1$, since from (b) these three vertices must have distinct colors.

As a conclusion, there is no bichromatic induced cycle with a vertex of $G-H$ and a vertex of $H_1$. Thus every minimum acyclic coloring of $G$ which satisfies (b) induces a minimum acyclic coloring of $G-H$.

By Lemma \ref{lemma-split}, if $G-H$ is not empty, then the characteristic graph of $H$ is a split graph ($H_1$ represents the clique and $H_2$ represents the independent set). Then, there is no induced $P_3$ with one vertex $u$ of $H_1$ and two vertices $w_1,w_2$ of $H_2$, since $w_1$ and $w_2$ must be in the same maximal homogeneous set, and therefore $u$ is adjacent to both or none of $w_1$ and $w_2$. By consequence, there is also no bichromatic induced $P_4$ with one vertex of $G-H$, one vertex $u$ of $H_1$ and two vertices $w_1,w_2$ of $H_2$ (otherwise we would have an induced $P_3$ $u-w_1-w_2$). Therefore, there is no bichromatic induced $P_4$ with a vertex of $G-H$ and a vertex of $H_1$. As a consequence, every minimum star coloring of $G$ which satisfies (b) induces a minimum star coloring of $G-H$.

Therefore, our work is to obtain a minimum acyclic (star) coloring of $G$ satisfying (b), given a minimum acyclic (star) coloring of $G-H$.
Let $\psi'_a$ and $\psi'_s$ be respectively an acyclic (star) coloring of $H$ such that the vertices of $H_1$ receive distinct colors.

Then we get respectively $k(\psi'_a)+\max\{0,\chi_a(G-H)-k_2(\psi'_a)\}$ and $k(\psi'_s)+\max\{0,\chi_{st}(G-H)-k_2(\psi'_s)\}$ colors for the acyclic-star colorings of $G$ respecting $\psi'_a$ and $\psi'_s$. Since $H$ has at most $q$ vertices, we can search in constant time the colorings $\psi'_a$ and $\psi'_s$ that minimize those values.

Now, assume (a). This case is easier than (b). Let $\psi_a$ and $\psi_s$ be respectively an acyclic-star coloring of $H$ (with no further restrictions). Then we get respectively $k(\psi_a)+\max\{0,|V(G-H)|-k_2(\psi_a)\}$ and $k(\psi_s)+\max\{0,|V(G-H)|-k_2(\psi_s)\}$ colors for the acyclic-star colorings of $G$ respecting (a), $\psi_a$ and $\psi_s$. Since $H$ has at most $q$ vertices, we can search in constant time the colorings $\psi_a$ and $\psi_s$ that minimize those values.

Finally, take the minimum between (a) and (b) to calculate $\chi_a$ and $\chi_{st}$.

To calculate $\pi(G)$, we just have to prove that every star coloring of $G$ has no repetition pattern with a vertex of $G-H$ and a vertex of $H$. To do this, consider by contradiction a star coloring of $G$ with a color repetition pattern $x_1\ldots x_p x_1\ldots x_p$ on vertices $v_1\ldots v_{2p}$ with vertices of $G-H$ and $H$. Without loss of generality, consider that $v_1\in G-H$. Let $v_k v_{k+1}$ be the first edge from $G-H$ to $H_1$. Clearly $k<p$ (otherwise, $v_{k+1-p}$ and $v_{k+1}$ have the same color $x_{k+1-p}$ and induce an edge, a contradiction).

Observe that $v_k$ and $v_{k+p}$ received color $x_k$, and $v_{k+p}$ and $v_{k+p+1}$ received color $x_{k+1}$.
If $v_{k+p}\in H_1$, then we have the edge $v_{k}v_{k+p}$ with colors $x_{k}$ and $x_{k}$, a contradiction.
If $v_{k+p+1}\in G-H$, then we have the edge $v_{k+1}v_{k+p+1}$ with colors $x_{k+1}$ and $x_{k+1}$, a contradiction.
If $v_{k+p}\in G-H$ and $v_{k+p+1}\in H_1$, then we have the bichromatic $P_4$ $v_k v_{k+1} v_{k+p}v_{k+p+1}$, a contradiction.
If $v_{k+p}\in H_2$ and $v_{k+p+1}\in H_1$, then, from Lemma \ref{lemma-split}, $v_{k+1}$ and $v_{k+p+1}$ are in the same maximal homogeneous set and, consequently, $v_{k+1} v_{k+p}$ is an edge and we have the bichromatic $P_4$ $v_k v_{k+1} v_{k+p}v_{k+p+1}$, a contradiction.

Finally, the last case we have to consider is when $v_{k+p}$ and $v_{k+p+1}$ are in $H_2$.
Observe that $v_{p+1},\ldots,v_{k+p+1}\in H_2$ (otherwise, we have some edge $v_{\ell}v_{p+\ell}$, $\ell\in\{1,\ldots,k\}$, whose vertices have the same color $x_{\ell}$, a contradiction) and are in the same maximal homogeneous set $M_2\subseteq H_2$ from Lemma \ref{lemma-split}.

Let $\ell\in\{1,\ldots,p\}$ be the minimum integer such that $v_{p+\ell}\not\in M_2$ (this integer must exist, otherwise $M_2$ contains all colors and we have some edge from $H_1$ to $M_2$ whose vertices have the same color, a contradiction). Let $M_1\subseteq H_1$ be the maximal homogeneous set which contains $v_{p+\ell}$. If $v_{\ell}\in G-H$, then we have the edge $v_{\ell}v_{p+\ell}$ whose vertices have the same color $x_{\ell}$, a contradiction. If $v_{\ell}\in H_1$, then $v_{\ell}v_{p+\ell}$ could not be an edge and, from Lemma \ref{lemma-split}, $v_{\ell}\in M_1$ and, consequently, we have the edge $v_{\ell}v_{p+\ell-1}$ and the bichromatic $P_4$ $v_{\ell-1}v_{\ell}v_{p+\ell-1}v_{p+\ell}$, a contradiction.

If $v_{\ell}\in H_2$, let $t<\ell$ be the maximum integer not in $H_2$ smaller than $\ell$.
Clearly $t$ exists and $v_{t+1},\ldots,v_{\ell}\in H_2$ and are in the same maximal homogenoeus set $M_2'\subseteq H_2$ from Lemma \ref{lemma-split}. If $v_t\in M_1$, then we have the edge $v_t v_{p+t}$ whose vertices have the same color $x_t$, a contradiction. If $v_t\not\in M_1$, then, from Lemma \ref{lemma-split}, $v_t v_{p+\ell}$ is an edge and we have the bichromatic $P_4$ $v_{\ell}v_t v_{p+\ell}v_{p+t}$, a contradiction.

\end{proof}

\subsection{Harmonious coloring}

\begin{proof}[Proof of Lemma \ref{lema-spid-h}]
Consider a harmonious coloring $C$ of $G$. If $G$ is the join of two graphs $G_1$ and $G_2$, then there is no two vertices $x$ and $y$ of $G_1$ with the same color, since they have a common neighbor in $G_2$. Similarly, we have the same for $G_2$. Since no color of $G_1$ can appear in $G_2$, we have that all vertices must have distinct colors in $C$.

Now consider $G$ a spider with partition $(S,C,R)$.
Consider an harmonious coloring $C$ of $G$.
Since $C$ induces a clique, all vertices of $C$ must receive different colors.
Since every vertex of $R$ is adjacent to every vertex of $C$, then no color of $C$ can occur in $R$ (and vice-versa).
If $|R|>1$, then two vertices of $R$ cannot have the same color, because they are adjacent to the same vertex in $C$ (otherwise, the coloring is not harmonious). Summarizing, 
no two vertices of $R\cup C$ are assigned to the same color. 

Consider now that $G$ is a thin spider.
Let $s_i$ be a vertex of $S$. Let $c_i$ be the vertex of $C$ that is adjacent to $s_i$.
The color of $s_i$ cannot be the color of $c_j$ of $C$, otherwise $(c_i,s_i)$ and $(c_i,c_j)$ would be edges with the same pair of color in their endpoints. For the same reason, the color of $s_i$ cannot be the color of any vertex $r_j$ of $R$. So, a new color, distinct from those used to color $R \cup C$, have to be used to color $S$. 
Then, $h(G)\leq|R|+|C|+1$.

Observe that by coloring every vertex of $S$ with this new color would not produce two edges with the same pair of colors in their endpoints, because this new color appears only in $S$ and there is a bijection between vertices of $S$ and $C$, i.e., each vertex of $S$ is connected to exactly one vertex of $C$ different from the others. Therefore, $h(G)\geq|R|+|C|+1$.

Consider now that $G$ is a thick spider with $|C|>2$. 
Again, as before, the colors used in $C \cup R$ cannot be reused to color the vertices of $S$. 
It remains to know if two vertices $s_i$ and $s_j$ of $S$ can receive the same color.
As $|C|>2$, there is a vertex $c_k$ of $C$ that is adjacent to $s_i$ and $s_j$.
Then, $s_i$ and $s_j$ must receive different colors on a harmonious coloring.
It means that all the vertices of $S$ have to be assigned to distinct colors, that is, $h(G)=|R|+|C|+|S|$.
As $|S|=|C|$, we have the result.
\end{proof}

\begin{proof}[Proof of Lemma \ref{teo-small-h}]
Let $c$ be an harmonious coloring of $G$.
Let $c_H$ be the restriction of $c$ to $H$.
Observe that as $c$ is harmonious, then $c_H$ is harmonious too. 
Therefore, every harmonious coloring of $G$ can be obtained from some harmonious coloring of $G[H]$.

Note that there is no two vertices $x$ and $y$ of $G-H$ with the same color in $c$, since $x$ and $y$ have a common neighbor in $H_1$.
Hence, all vertices of $G-H$ receive distinct colors.

Consider then a harmonious coloring $c_H$ of $G[H]$. We are looking for a harmonious coloring $c_G$ of $G$ that extends $c_H$ and uses a minimum number of colors. 
It is necessary to know which colors of $c_H$ can be used to color the vertices in $G-H$.
As all the vertices in $G-H$ are adjacent to all the vertices of $H_1$, then the colors that can be used are the ones in $c_H$ such that no vertex with this color is adjacent to other vertex colored with some color that appears in $H_1$.

Summarizing, let $c_1$ be the set of the colors in $c_H$ of the vertices in $H_1$.
Let $X$ be the subset of vertices in $H$ colored with a color of $c_1$.
Clearly, $H_1\subseteq X$.
Let $Y=X\cup N(X)$, where $N(X)$ is the set of the neighbors of the vertices in $X$.
Let $c_Y$ be the set of colors in $c_H$ with a vertex of $Y$.
Let $c_Z$ be the set of colors in $c_H$ that are not used in $Y$. That is, $c_Y\cup c_Z=c_H$ and $c_Y\cap c_Z=\emptyset$.

Clearly, we cannot use any color $\gamma$ from $c_Y$ in $G-H$, because either $\gamma$ is a color of $H_1$ or it exists already a vertex of $H$ with the color $\gamma$ which is
neighbor from some vertex in $H_1$. It is easy to see that we can use the colors of $c_Z$ in $G-H$.
With this, a harmonious coloring $c_G$ of $G$ that extends $c_H$ and use a smallest number of colors must use $|c_H|+\max\{|V(G)|-|H|-|c_Z|,0\}$ colors and it is obtained  by coloring each vertex of $G-H$ either using colors of $c_Z$ or adding new colors with respect to the colors used in $c_H$.

As the size of $H$ is less than $q$ (which is a constant independent from the size of $G$), we can get in constant time all the harmonious colorings $c_H$ of $H$ and calculate the minimum number of colors in a harmonious coloring of $G$ that extends $c_H$. By doing this, we get in constant time the harmonious chromatic number of $G$ and we can obtain on linear time in the number of vertices a harmonious coloring of $G$.
\end{proof}

\subsection{Clique coloring}

\begin{proof}[Proof of Lemma \ref{lema-spid-c}]
The proof is direct if $G=G_1\cup G_2$.
If $G$ is the join of two graphs $G_1$ and $G_2$, then it is easy to see that every maximal clique of $G$ must have vertices of $G_1$ and $G_2$, since, for every clique $C$ of $G_1$, $C\cup\{v_2\}$ (where $v_2\in G_2$) is a clique of $G$. Then, coloring the vertices of $G_1$ with color 1 and the vertices of $G_2$ with color 2, we have that every maximal clique receives two colors. 

Now suppose that $G$ is a quasi-spider with partition $(R,C,S)$. Suppose first that $R$ is not empty. The same argument below shows that there is no maximal clique of $G$ with vertices of $R$ and no vertex of $C$ and that there is no maximal clique of $G$ with vertices of $C$ and no vertex of $R$ or no vertex of $S$. Since there is no clique with two vertices of $S$, we can obtain a clique coloring of $G$ by coloring the vertices of $R$ and $S$ with color 1 and the vertices of $C$ with color 2.

Suppose now that $R$ is empty. In this case, is is possible that $C$ is a maximal clique.
Let $H=K_2$ or $H=\overline{K_2}$ be the subgraph that replaced a vertex of $C\cup S$ in the definition of quasi-spider. Let $x\in C-H$ be a vertex of $C$ that is not in $H$. Let $N(x)$ be the set of neighbors of $x$ in $S$. It is easy to see that coloring $C-\{x\}$ and $N(x)$ with color 1 and $x$ and $S-N(x)$ with color 2, we have that every maximal clique receives two colors. Then we have a 2-clique-coloring of $G$.
\end{proof}

\begin{proof}[Proof of Lemma \ref{teo-small-c}]
At first, suppose that $G-H$ is not empty. Then $H$ is a separable p-component and, by Lemma \ref{lemma-split}, the characteristic graph of $H$ is a split graph ($H_1$ ''\emph{reduces}'' to a clique and $H_2$ ''\emph{reduces}'' to an independente set). Also remember that every vertex of $H_2$ has a neighbor in $H_1$. Hence, if two vertices of $H_2$ induces an edge, then they are in the same homogeneous set and then they have a common neighbor in $H_1$. Consequently, there is no maximal clique with vertex set contained in $H_2$.

It is easy to see that there is no maximal clique of $G$ with vertices of $G-H$ and no vertex of $H_1$, since, for every clique $C$ of $G-H$, $C\cup\{v\}$ (where $v\in H_1$) is a clique of $G$. The same argument shows that there is no maximal clique of $G$ with vertices of $H_1$ and no vertex of $H_2$ or no vertex of $G-H$.Then, we can obtain a clique coloring of $G$ by coloring the vertices of $G-H$ and $H_2$ with color 1 and the vertices of $H_1$ with color 2.

Now, suppose that $G-H$ is empty. Since $H$ is a p-connected $(q,q-4)$-graph, then $H$ has at most $q-1$ vertices. Since $q$ is fixed, we can generate all possible clique-colorings in constant time and obtain the clique chromatic number.
\end{proof}


\end{document}